\documentclass[11pt, letterpaper]{article} 

\usepackage[margin=1in]{geometry} 
\usepackage{times} 
\usepackage{amsmath, amssymb, amsthm} 
\usepackage{graphicx} 
\usepackage[english]{babel} 
\usepackage[utf8]{inputenc} 
\usepackage{parskip} 
\newtheorem{theorem}{Theorem}[section] 
\newtheorem{lemma}[theorem]{Lemma}     
\newtheorem{corollary}[theorem]{Corollary}

\usepackage{natbib}
\usepackage[colorlinks=true, allcolors=blue]{hyperref}

\usepackage{graphicx}
\usepackage{enumerate}
\usepackage{authblk}
\usepackage{color,soul}
\usepackage{algorithm}
\usepackage{algorithmic}
\usepackage{caption}
\usepackage{subcaption}
\usepackage{tikz}

\newcommand{\Q}{\mathcal{Q}}
\newcommand{\A}{\mathcal{A}}
\newcommand{\B}{\mathcal{B}}

\newcommand{\E}{\mathbb{E}}

\newcommand{\I}{\mathcal{I}}

\usepackage{times}

\title{\textbf{Learning-augmented online caching: new upper bounds}} 
\author[1,2]{Daniel Skachkov}
\author[3]{Denis Ponomaryov}
\author[1,2]{Yuriy Dorn}
\author[3]{Alexander Demin}
\affil[1]{Institute for Artificial Intelligence, Lomonosov Moscow State University, Moscow, Russia}
\affil[2]{Moscow Institute of Physics and Technology, Moscow, Russia}
\affil[3]{Ershov Institute of Informatics Systems, Novosibirsk, Russia}
\date{} 

\begin{document}

\maketitle



\section{Introduction}

The topic of storage management has motivated quite a number of novel methods in online optimization. 
According to the existing classification (\cite{supersurvey}), storage management algorithms are subdivided into several groups including replacement, admission, prefetching, and allocation policies. In this paper, we focus on replacement policies (also called \textit{eviction strategies}).

A \textit{cache hit} occurs when the requested element (e.g., page) is already present in the cache. Conversely, a \textit{cache miss} happens when the requested page is not in the cache. It then needs to be loaded into the cache. If the cache has reached its capacity, an eviction strategy must decide which page to evict in order to create space for the new one. The objective of this strategy is to minimize the total number of cache misses by managing evictions efficiently.

It is known that in offline case, when the whole sequence of requests is known in advance, Belady algorithm (\cite{Belady}) provides an optimal solution. This algorithm evicts page with the largest distance to its next request. The efficiency of other replacement policies is typically measure against that of Belady in terms of competitive ratio.

\textit{Competitive ratio} of an eviction strategy $\mathcal{S}$ is defined as the ratio of the number of cache misses achieved by $\mathcal{S}$ to the best possible number (i.e., achieved by Belady) in the worst case, which means that we consider the maximum of this value over all possible inputs. In general, when the strategy is randomized, we consider the expected number of cache misses, where the expectation is taken over the random bits of the algorithm.

Different replacement strategies have been proposed in the literature including the classical LFU and LRU algorithms. LFU (Least Frequently Used) evicts a page which is requested less frequently than other pages, while LRU (Least Recently Used) evicts the page that was not requested for the longest time. Also mixtures of well-known replacement strategies have been proposed in the literature with the goal to more flexibly support varying data access patterns (we refer the reader to survey \cite{ReplacementPoliciesBook}).

The theoretical lower bound on the competitive ratio of eviction strategies is known to be \( \Omega(k) \) for deterministic strategies and \( \Omega(\log k) \) for randomized (stochastic) strategies (\cite{RandomBook}). These bounds are tight. LRU is known to be optimal among deterministic strategies, while one of the best known theoretically optimal randomized replacement strategies is the Marker algorithm (\cite{MarkerAlgorithm}).   

The Marker algorithm consists of a series of so-called epochs. At the very beginning of each epoch, all pages in the cache are considered as unmarked. When a cache hit occurs, the corresponding page is marked. In the case of a cache miss, an unmarked page is selected uniformly at random and it is evicted, after which the new page is marked immediately. When there are no more unmarked pages left, the current stage is over and the algorithm proceeds to a new stage.

Lykouris and Vassilvitskii  (\cite{PredictiveMarker}) suggested considering a variant of the cache replacement problem, the so-called \textit{learning-augmented} online caching. In this setting, each request arrives with a prediction of its next occurrence time, which is the time when the same page will be requested in the future. Specifically, they considered a Predictive Marker algorithm. It aims at balancing between the pure Marker algorithm and its modified version, which chooses not an arbitrary unmarked page, but the one with the highest predicted request time. This balancing allows not only to achieve a constant competitive ratio for sufficiently accurate predictions, but also guarantees \( O(\log k) \) competitive ratio in the worst case. 

In \cite{MTS} Metric Task Systems (MTS) was introduced  as a class of online convex optimization problems. An MTS instance is given by a metric space \( S \), with elements being states, and a sequence of tasks to be completed. The cost of executing a task differs from task to task and depends on the current state. The cost of switching between states is defined by the distance matrix of space \( S \). The goal is to switch between states in such a way that minimizes the total cost. Notice that if we take \( S \) as consisting of all possible cache states, with distances equal to zero, and assume that page requests are tasks of cost \( 1 \) if there is a cache miss and \( 0 \) otherwise, we obtain the cache replacement problem. 

In \cite{MTS_ML} authors consider MTS with predictions, which are elements of a metric space \( S \) and ideologically predict the behavior of some offline algorithm (not necessarily optimal). Given any \( \alpha \)-competitive online algorithm, they showed how to improve it by using  predictions, so that it is still $\alpha$-competitive in the worst case, but tends to the results of an offline algorithm with the decrease of prediction error. 

Alexander Wei (\cite{PredictiveMarkerImprovement2}) obtained upper bounds for the competitive ratio of the \textsc{BlindOracle} algorithm, which operates like the Belady algorithm by using predictions (i.e., evicting the page with the largest predicted next request time). Also, there were presented upper bounds for combinations of \textsc{BlindOracle} with LRU and Marker. Rohatgi Dhruv established a lower bound on the competitive ratio of any randomized algorithm (\cite{PredictiveMarkerImprovement1}).

The paper is organized as follows. 
In Section~\ref{Sect:Background}, we review known upper and lower bounds on the competitive ratio for \textsc{BlindOracle}, randomized algorithms, and their combinations. 
Section~\ref{Sect:Contribution} summarizes the main contributions of this work. 
Section~\ref{Sect:Statement} introduces the necessary notation and formalizes the problem statement. 
The description of our randomized algorithm is presented in Section~\ref{Sec:Alg}. 
Section~\ref{Sect:EvictionGraph} discusses eviction graphs and their existence—an auxiliary tool for proving competitive ratio bounds. 
This section also derives new upper bounds for \textsc{BlindOracle} using this technique. 
In Section~\ref{Sect:UpperBound}, we prove the upper bound for our algorithm by leveraging the framework from Section~\ref{Sect:EvictionGraph}. 
Finally, Section~\ref{Sect:Conclusion} concludes the paper.

\section{Related Work}\label{Sect:Background}
It is known that online algorithms for cache replacement can be combined in a way that ensures that the resulting algorithm performs nearly as well as the best of the combined ones. For a strategy $\mathcal{S}$, let $\mathrm{OBJ}_\mathcal{S}(\sigma)$ denote the objective function value (i.e., the number of cache misses) produced by $\mathcal{S}$ for some sequence of requests $\sigma$. We also use the notation $\mathrm{OPT}(\sigma)$ as a shorter way to represent $\mathrm{OBJ}_{\mathcal{A}}(\sigma)$, where $\mathcal{A}$ is an optimal offline algorithm (Belady). For simplicity, we will usually omit $\sigma$ and use notations $\mathrm{OBJ}_\mathcal{S}$ and $\mathrm{OPT}$ instead, unless we want to specify the exact sequence we are working with. In \cite{MarkerAlgorithm} the following result was obtained: 

\begin{theorem}[\cite{MarkerAlgorithm}, special case] \label{theorem_combining_deterministic}
    Let $\mathcal{S}_1$ and $\mathcal{S}_2$ be eviction strategies. Then there exists an eviction strategy $\mathcal{S}$ such that
    $$
    \mathrm{OBJ}_\mathcal{S}(\sigma) \leq 2 \min\{ \mathrm{OBJ}_{\mathcal{S}_1} (\sigma), \mathrm{OBJ}_{\mathcal{S}_2} (\sigma)\} + O(1)
    $$
    for any sequence $\sigma$.
     Moreover, if both $\mathcal{S}_1$ and $\mathcal{S}_2$ are deterministic, then so is $\mathcal{S}$.
\end{theorem}

They specifically showed that the desired strategy $\mathcal{S}$ can be constructed by simulating the execution of $\mathcal{S}_1$ and $\mathcal{S}_2$ independently and evicting any page that is absent in the cache of the most successful strategy.

A more sophisticated technique which involves stochastically switching between strategies was used in \cite{combining} to prove the following theorem:

\begin{theorem}[\cite{combining}, special case] \label{theorem_combining_stohastic}
    Let $\mathcal{S}_1$ and $\mathcal{S}_2$ be eviction strategies. Then for any $\gamma \in (0, \frac{1}{4})$, there exists a randomized strategy $\mathcal{S}$ such that for any sequence $\sigma$ holds:
    $$
    \mathrm{OBJ}_\mathcal{S}(\sigma) \leq (1 + \gamma) \min\{\mathrm{OBJ}_{\mathcal{S}_1}(\sigma), \mathrm{OBJ}_{\mathcal{S}_2}(\sigma)\} + O(\gamma^{-1}k).
    $$

\end{theorem}
For the Predictive Marker algorithm  (\cite{PredictiveMarker}) it is known that its competitive ratio is bounded by
$$
2 + O\Big(\min\Big\{\log k, \sqrt{\frac{\eta}{\mathrm{OPT}}}\Big\}\Big),
$$
where $\eta$ is $L_1$ norm of the prediction loss (its formal definition will be given later).

A modified version of the Marker algorithm proposed in \cite{PredictiveMarkerImprovement1} has a competitive ratio  
$$
O\Big(1 + \min\Big\{\log k, \frac{\log k}{k} \frac{\eta}{\mathrm{OPT}}\Big\}\Big).
$$
In \cite{PredictiveMarkerImprovement2} it was shown that the \textsc{BlindOracle}, which operates like the Belady algorithm using predictions (i.e., evicting the page with the largest predicted next request time), has a competitive ratio at most
$$
\min\Big\{1 + 2\frac{\eta}{\mathrm{OPT}}, 4 + \frac{4}{k - 1} \frac{\eta}{\mathrm{OPT}}\Big\}.
$$
Combining \textsc{BlindOracle} with LRU (Theorem \ref{theorem_combining_deterministic}) allows to achieve the competitive ratio of 

$$
2 \min \Big\{k, \min\Big\{1 + 2\frac{\eta}{\mathrm{OPT}}, 4 + \frac{4}{k - 1} \frac{\eta}{\mathrm{OPT}}\Big\}\Big\},
$$
(which is also asymptotically optimal among deterministic strategies, as was shown in \cite{PredictiveMarkerImprovement2}) while combining \textsc{BlindOracle} with Marker  (Theorem \ref{theorem_combining_stohastic}) provides the competitive ratio bounded by

$$
(1 + \gamma)\min\Big\{H_k, \min\Big\{1 + 2\frac{\eta}{\mathrm{OPT}}, 4 + \frac{4}{k - 1} \frac{\eta}{\mathrm{OPT}}\Big\}\Big\}.
$$
Here $H_k = 1 + \frac{1}{2} + \frac{1}{3} + \ldots + \frac{1}{k} = \ln(k) + O(1)$ --- $k$-th harmonic number.

Finally \cite{ExperimentalComparing} is a practical comparative study of the mentioned algorithms (\cite{PredictiveMarker, PredictiveMarkerImprovement1, PredictiveMarkerImprovement2, MTS_ML}), and \textsc{BlindOracle} demonstrates one of the best results.

\section{Our Contribution}\label{Sect:Contribution}
First, we provide a more detailed analysis on the upper bound on the competitive ratio of the \textsc{BlindOracle} algorithm and improve it from
$$
\min\Big\{1 + 2\frac{\eta}{\mathrm{OPT}}, 4 + \frac{4}{k - 1} \frac{\eta}{\mathrm{OPT}}\Big\},
$$
presented in   \cite{PredictiveMarkerImprovement2}, to
$$
\min\Big\{1 + \frac{\eta}{\mathrm{OPT}}, 3 + \frac{3}{k} \frac{\eta}{\mathrm{OPT}}\Big\}.
$$
This refinement, when combined with Theorem \ref{theorem_combining_deterministic} and Theorem \ref{theorem_combining_stohastic}, implies that there exists a deterministic algorithm with the competitive ratio
$$
2 \min \Big\{k, \min\Big\{1 + \frac{\eta}{\mathrm{OPT}}, 3 + \frac{3}{k} \frac{\eta}{\mathrm{OPT}}\Big\}\Big\},
$$
and a randomized algorithm with the competitive ratio
$$
(1 + \gamma)\min\Big\{H_k, \min\Big\{1 + \frac{\eta}{\mathrm{OPT}}, 3 + \frac{3}{k} \frac{\eta}{\mathrm{OPT}}\Big\}\Big\}
$$
for any $\gamma \in (0, \frac{1}{4})$. 

Furthermore, we introduce a new randomized algorithm, referred to as the \textsc{AlternatingOracle} (Algorithm \ref{alg:AlternatingOracle}). A detailed description of this algorithm is provided in Section \ref{Sec:Alg}. We demonstrate that its competitive ratio is $O\left(\sqrt{\frac{1}{k} \cdot \frac{\eta}{\mathrm{OPT}}}\right)$. When combined with the \textsc{PredictiveMarker} algorithm (as stated in Theorem \ref{theorem_combining_stohastic}), this results in a competitive ratio of $O\left(\min\left\{\sqrt{\frac{1}{k} \cdot \frac{\eta}{\mathrm{OPT}}}, \log k\right\}\right)$.

\section{Problem statement} \label{Sect:Statement}

A cache (or buffer state) $\mathcal{C}_\mathcal{Q} \subset \mathcal{S}$ is a subset of elements called \textit{pages} from a given set $\mathcal{S}$. Here, $\mathcal{Q}$ denotes the eviction strategy managing the cache. The size of a cache $k\geq 1$ is the maximal number of pages it can contain, i.e. $|\mathcal{C}_\Q| \leq k$. Let $\sigma$ be a sequence of \textit{requests} of length $T$. We denote its elements as $\sigma(1), \sigma(2), \ldots, \sigma(T)$, where $\sigma(t) \in \mathcal{S}$ is a page requested at time step $t$. Each request refers to a single page. 

We also introduce the notation
$$
\nu(t) := \min \{s > t \ | \ \sigma(s) = \sigma(t)\},
$$
where $\nu(t) \in [1, T]$ represents the time of the next request of page $\sigma(t)$. In case the minimum is undefined, we set $\nu(t)$ to $T + 1$.

In this paper, we focus on algorithms equipped with an oracle. For each request $\sigma(t)$, the oracle provides a prediction $\omega(t)$ that estimates the value $\nu(t)$. We define an oracle loss on request $\sigma(t)$ as $\eta(t) := |\nu(t) - \omega(t)|$ and the total loss $\eta := \sum_{t = 1}^T \eta(t) = \sum_{t=1}^T |\nu(t) - \omega(t)|$.

At each time step $t \in [1, T]$ any \textit{feasible} algorithm $\mathcal{Q}$ observes the current request $\sigma(t)$, oracle prediction $\omega(t)$, and buffer state $\mathcal{C}_\Q(t-1)$, evicts page $p_\Q(t) = \mathcal{Q}\left (\sigma(1), \omega(1), \dots, \sigma(t), \omega(t), \mathcal{C}_\Q(t-1) \right)$, and updates cache $\mathcal{C}_{\Q}(t) = \sigma(t) \cup \left (\mathcal{C}_{\Q}(t-1) \setminus \{p_\Q(t)\} \right )$. 

The objective of the algorithm is to minimize the total number of cache misses:
$$
\mathrm{OBJ}_{\mathcal{Q}}(\sigma) = \sum_{t=1}^T \mathrm{I}\{\sigma(t) \notin \mathcal{C}_{\mathcal{Q}}(t-1)\}.
$$
We denote by $\mathrm{OPT}$ for given input $(\sigma, \omega)$ the total number of hit misses produced by Belady's algorithm (which is offline):
$$
\mathrm{OPT}(\sigma, \omega) = \sum_{t=1}^T \mathrm{I}\{\sigma(t) \notin \mathcal{C}_{Belady}(t - 1)\},
$$
which serves as a lower bound for any feasible algorithm. 

We say that an algorithm $\mathcal{Q}$ is \emph{$c$-competitive} if there exists a constant $c > 1$ (called the \textit{competitive ratio}) such that, for any input $(\sigma, \omega)$, the following inequality holds:
\[
\mathrm{OBJ}_{\mathcal{Q}}(\sigma, \omega) \leq c \cdot \mathrm{OPT}(\sigma, \omega).
\]
For randomized algorithms, we consider the expected value $\mathbb{E}[\mathrm{OBJ}_\mathcal{Q}]$ instead of $\mathrm{OBJ}_\mathcal{Q}$, where the expectation is taken over the internal randomness of the algorithm. In this case, we evaluate the average performance of the algorithm on any given input.

Since we reference pages by their indices, in addition to maintaining $\mathcal{C}_\mathcal{Q}(t)$, we also keep track of a set $\mathcal{I}_\mathcal{Q}(t)$, defined as:
\[
\mathcal{I}_\mathcal{Q}(t) := \left\{ \max \{t' \leq t \mid \sigma(t') = s\} \mid s \in \mathcal{C}_\mathcal{Q}(t) \right\}.
\]
In other words, $\mathcal{I}_\mathcal{Q}(t)$ contains the most recent time step of the last request for each page currently in the cache. 

There are two scenarios in which an index $i$ belongs to $\mathcal{I}_\mathcal{Q}(t) \setminus \mathcal{I}_\mathcal{Q}(t+1)$:
\begin{itemize}
    \item The page $\sigma(i)$ was evicted while processing the request $\sigma(t+1)$, i.e., $p_\Q(t+1) = \sigma(i)$, or
    \item The request $\sigma(t+1)$ corresponds to the same page as $\sigma(i)$, and the index $i$ is replaced by $t+1$.
\end{itemize}

Finally, $i \in \mathcal{I}_\mathcal{Q}(t+1) \setminus \mathcal{I}_\mathcal{Q}(t)$ if and only if $i = t + 1$.

Furthermore, when we say that a page $\sigma(i)$ is stored in the cache at time step $t$, we typically imply not only that $\sigma(i) \in \mathcal{C}_\mathcal{Q}(t)$, but also that $i \in \mathcal{I}_\mathcal{Q}(t)$.

\section{Algorithm}\label{Sec:Alg}

Next, we present our algorithm (Algorithm~\ref{alg:AlternatingOracle}). The core idea is to combine three distinct eviction strategies. The first is \textsc{BlindOracle}~\cite{PredictiveMarkerImprovement2}, which evicts the page with the largest predicted value $w(s)$. The second strategy, \textsc{RandomAlg}, performs a random eviction by selecting a uniformly random page from the cache. The third strategy, \textsc{Corrector}, targets pages whose predicted next request time has already expired; that is, at time $t$, it evicts a page $\sigma(s)$ with $w(s) < t$. If there are multiple such pages, the algorithm evicts the one with the largest prediction error, i.e., the page maximizing $t - \omega(s)$ over $s \in \I(t)$, which is equivalent to minimizing $\omega(s)$. If no such pages exist, \textsc{Corrector} defaults to using \textsc{BlindOracle}.

We note that, for the purposes of our analysis, it would be sufficient for \textsc{Corrector} to always evict the page $\sigma(s)$ with the minimal value of $\omega(s)$, without reverting to \textsc{BlindOracle}. However, we chose not to do this, as such behavior may appear counterintuitive.

The rationale for combining these strategies is as follows. The \textsc{BlindOracle} performs optimally among deterministic algorithms and is effective in typical scenarios. The \textsc{RandomAlg} strategy is designed to address adversarial inputs of the form $\sigma(t+1) = p(t)$, where a page evicted in the previous step is immediately requested again—an approach that can defeat any deterministic algorithm. Finally, the \textsc{Corrector} strategy is intended to handle pages with certified prediction errors, ensuring that such inaccuracies are promptly corrected.

\begin{algorithm}[h] 
\caption{\textsc{AlternatingOracle}}
\label{alg:AlternatingOracle}
    \textbf{Input}: Request $\sigma(t)$ abscent in the cache, pages stored in the cache, evictions history\\
    \begin{algorithmic}[1] 
    \IF{page $\sigma(t)$ was not requested previously}
    \STATE \label{AlternatingOracle::line1} \textsc{BlindOracle}: Evict page with maximal value of $\omega$ in the cache
    \ELSE
    \IF{page $\sigma(t)$ was last evicted using \textsc{BlindOracle}}
    \STATE \label{AlternatingOracle::line2}  RandomAlg: Evict page chosen uniforly from the cache
    \ENDIF
    \IF{page $\sigma(t)$ was last evicted using RandomAlg}
    \STATE \label{AlternatingOracle::line3} Evict page using Corrector (Alg. \ref{alg:Corrector})
    \ENDIF
    \IF{page $\sigma(t)$ was last evicted using Corrector}
    \STATE \label{AlternatingOracle::line4} \textsc{BlindOracle}: Evict page with maximal value of $\omega$ in the cache
    \ENDIF
    \ENDIF
    \end{algorithmic}
\end{algorithm}

\begin{algorithm}[h] 
\caption{Corrector}
\label{alg:Corrector}
    \textbf{Input}: Request $\sigma(t)$ abscent in the cache, pages stored in the cache\\
    \begin{algorithmic}[1] 
    \STATE  \label{Corrector::line1} $W := \{\sigma(i) \ | \ \sigma(i) \text{ is in the cache and }\omega(i) < t \text{, i.e. prediction is provably wrong} \}$
    \IF{$W$ is not empty}
    \STATE  \label{Corrector::line2} Evict page $\sigma(i)$ from $W$ with minimal value of $\omega(i)$
    \ELSE
    \STATE \label{Corrector::line3} Evict page using \textsc{BlindOracle}
    \ENDIF
    \end{algorithmic}
\end{algorithm}

\begin{theorem} \label{Theorem_MainResult}
    The \textsc{AlternatingOracle} algorithm is $O\left(\sqrt{\frac{1}{k} \cdot \frac{\eta}{\mathrm{OPT}}}\right)$ -competitive.
\end{theorem}

We present a brief outline of our proof. Firstly, we introduce the concept of the eviction graph. Informally, the edges of this graph represent mistakes made by the eviction strategy, i.e., instances where the strategy deviates from the optimal offline algorithm (Belady's algorithm). Inspired by \cite{PredictiveMarkerImprovement2}, we establish an upper bound on the competitive ratio in terms of the number of edges in the eviction graph (Lemma \ref{Lemma_ErrorGraphExistence}). Then we introduce a more refined technique for constructing the eviction graph with colored edges (Lemma \ref{Lemma_ColoredErrorGraphExistence}). We demonstrate that the number of edges of each color is bounded by the oracle's loss (Lemma \ref{Lemma_LongChain}). We then sum up over all colors and use the fact that, from one side, the gap between algorithm performance and optimum is bounded by the total number of edges, while on the other side, the number of edges is bounded by the oracle's loss, which allows us to establish a connection between performance and oracle's loss (Lemma \ref{Lemma_MainResult}).

\section{Eviction graph}\label{Sect:EvictionGraph}

Suppose that after processing query $\sigma(t)$, algorithm $\mathcal{Q}$ evicts page $\sigma(i) = p_{\mathcal{Q}}(t)$. This is followed by a cache miss at time step $t' = \nu(i)$, i.e., $\sigma(t') = \sigma(i)$, and processing $\sigma(t')$ results in the eviction of page $\sigma(j) = p_{\mathcal{Q}}(t')$. In this case, we say that the eviction of $\sigma(i)$ \textit{triggers} the eviction of $\sigma(j)$.

We say that a tuple $(s_1, \dots, s_d)$ is a \textit{chain} if for each $i$, $1 \leq i \leq d-1$, the eviction of $\sigma(s_i)$ triggers the eviction of $\sigma(s_{i+1})$ (see Fig.~\ref{fig:ChainExample}).

\begin{figure}[h]
    \centering
    \includegraphics[width=0.95\textwidth]{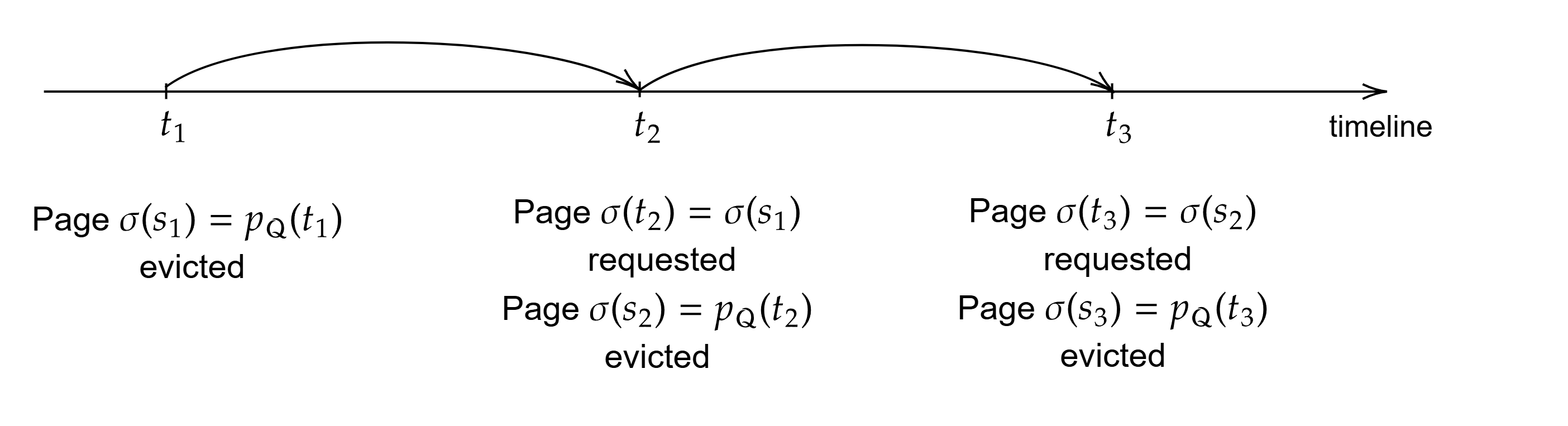}
    \caption{An example of chain $(s_1, s_2, s_3)$.}
    \label{fig:ChainExample}
\end{figure}

We also say that evictions of pages $p_\Q(t)$ and $p_\Q(t')$, $t < t'$, do not \textit{overlap}, if eviction of page $p_\Q(t)$ triggered eviction of page $p_\Q(\tilde{t})$ and $\tilde{t} \leq t'$. We note that according to this definition all evictions in the chain do not overlap each other.

Consider some arbitrary algorithm $\mathcal{Q}$ at the arrival of request $\sigma(t+1)$. Let $\sigma(i)$ and $\sigma(j)$ be in the cache, i.e. $i, j \in \mathcal{I}_\Q(t)$. We say that page $\sigma(i)$ is \textit{better} (as eviction candidate) than $\sigma(j)$, if $\nu(i) > \nu(j)$. We say that page $\sigma(i)$ is optimal if it is not worse than any other page stored in the cache. Notice that according to this definition, the optimal offline algorithm (Belady) always evicts only optimal pages. We say that algorithm $\Q$ \textit{made a mistake} if it evicted a non-optimal page.

Before formulating the following lemma, we need to define the \textit{right} and \textit{left} losses for each prediction $\omega(t)$:
\[
\eta^+(t) := \max\{\omega(t) - \nu(t), 0\}, \qquad
\eta^-(t) := \max\{\nu(t) - \omega(t), 0\}.
\]
Clearly, $\eta(t) = \eta^+(t) + \eta^-(t)$.

\begin{lemma} \label{Lemma_LeftRightError}
    Let $\mathcal{B}$ denote the \textsc{BlindOracle} algorithm. Suppose that, when processing the request $\sigma(t)$, algorithm $\mathcal{B}$ makes a mistake by evicting a non-optimal page $\sigma(j)$, even though there exists a better page $\sigma(i)$ in the cache with $\nu(i) > \nu(j)$. Then the following inequality holds:
    \[
        \eta^-(i) + \eta^+(j) \geq \nu(i) - \nu(j).
    \]
\end{lemma}

\begin{proof}
    Since algorithm $\mathcal{B}$ preferred to evict $\sigma(j)$ over $\sigma(i)$, by the definition of \textsc{BlindOracle}, we have $\omega(j) \geq \omega(i)$. Observe that
    \[
        \eta^-(i) \geq \nu(i) - \omega(i), \qquad
        \eta^+(j) \geq \omega(j) - \nu(j).
    \]
    Summing these inequalities, we obtain
    \[
        \eta^-(i) + \eta^+(j) \geq (\nu(i) - \omega(i)) + (\omega(j) - \nu(j)) = \nu(i) - \nu(j) + (\omega(j) - \omega(i)) \geq \nu(i) - \nu(j),
    \]
    where the last inequality follows from $\omega(j) \geq \omega(i)$.
\end{proof}

We utilize the notion of an \textit{eviction graph} generated by an arbitrary algorithm $\mathcal{Q}$ while processing the input sequence $\sigma$. An oriented graph $H_{\mathcal{Q}}$ is called an \textit{eviction graph} if it satisfies the following conditions:
\begin{itemize}
    \item Its vertex set is $\{1, 2, \ldots, T\}$;
    \item For any edge $(i, j)$, there exists a time $t+1$ such that algorithm $\mathcal{Q}$ made a mistake by evicting page $\sigma(j)$, even though page $\sigma(i)$ with $\nu(i) > \nu(j)$ was a better candidate to evict. Furthermore, we require that each mistake induces at most one edge: that is, if there are two distinct pages $\sigma(i)$ and $\sigma(i')$ both better than $\sigma(j)$, then at most one of the edges $(i, j)$ or $(i', j)$ may be present in $H_{\mathcal{Q}}$.
\end{itemize}

The eviction graph is defined ambiguously. For example, deleting any of its edges preserves satisfaction of the constraints. Moreover, if algorithm $\mathcal{Q}$ is stochastic, the set of feasible graphs depends on random bit realizations. Therefore, when we state that an eviction graph $H_{\mathcal{Q}}$ exists satisfying certain conditions, we mean that for every random bit realization, such a graph can be constructed (potentially different for each realization).

\begin{lemma} \label{Lemma_ErrorGraphForest}
    Let $\Q$ be an arbitrary algorithm, and let $H_\Q$ be its eviction graph. Then $H_\Q$ is an oriented forest.
\end{lemma}
\begin{proof}
    We need to show two properties: (1) $H_\Q$ contains no cycles, and (2) the in-degree of any vertex is at most one. These two properties together imply that $H_\Q$ is an oriented forest.

    First, consider the absence of cycles. By the construction of the eviction graph, if there is an edge $(i, j)$, then $\nu(i) > \nu(j)$. Therefore, along any path in the graph, the value of $\nu$ strictly decreases, preventing cycles.

    Second, consider the in-degree of each vertex. By definition, for each mistake (i.e., each time a non-optimal page is evicted), at most one edge is added to the graph for that event. Furthermore, after a page $\sigma(j)$ is evicted, the corresponding index $j$ is permanently removed from the set $\I_\Q$. Therefore, no further edges can be directed into $j$ after its eviction. Thus, the in-degree of any vertex is at most one.

    Together, these two properties ensure that $H_\Q$ is an oriented forest.
\end{proof}

\begin{lemma} \label{Lemma_EdgeProperties}
    Let $\Q$ be any algorithm, and let $H_\Q$ be its eviction graph. If there is an edge $(i, j)$ in $H_\Q$, then $\nu(j) \leq T$; that is, $\sigma(j)$ is not the last request for page $\sigma(j)$. In particular, this means that for any two distinct edges $(i, j_1)$ and $(i, j_2)$, we have $\nu(j_1) \neq \nu(j_2)$.
\end{lemma}

\begin{proof}
    By definition, the page $\sigma(j)$ was mistakenly evicted by algorithm $\Q$, even though page $\sigma(i)$ would have been a better choice. This means that $\nu(i) > \nu(j)$. Since $\nu(i) \leq T + 1$, it follows that $\nu(j) \leq T$. As a result, for any two distinct edges $(i, j_1)$ and $(i, j_2)$, we must have $\nu(j_1) \neq \nu(j_2)$.
\end{proof}

Next, we introduce Lemma \ref{Lemma_ErrorGraphExistence}, which is a slightly adjusted version of the supporting lemma from \cite{PredictiveMarkerImprovement2} and the proof follows \cite{PredictiveMarkerImprovement2}.

\begin{lemma} \label{Lemma_ErrorGraphExistence}
    For any algorithm $\Q$, there exists an eviction graph $H_\Q$ such that
    $$
    \mathrm{OBJ}_\Q \leq \mathrm{OPT} + |E|,
    $$
    where $E$ is the set of edges in $H_\Q$.
\end{lemma}

\begin{proof}
Let $\A$ denote the optimal offline algorithm (Belady), so that $\text{OBJ}_\A = \text{OPT}$. We will simulate the execution of both algorithms $\A$ and $\Q$ in parallel, each maintaining its own cache.

We construct a bipartite graph $X(t)$ at each time $t$, with partitions $\I_\A(t)$ and $\I_\Q(t)$. This graph satisfies two invariants:
\begin{itemize}
    \item Its edges are disjoint (no two edges share a vertex).
    \item If there is an edge $(i, j)$, then $\nu(i) \geq \nu(j)$.
\end{itemize}

Consider the moment when request $\sigma(t+1)$ arrives after processing request $\sigma(t)$. We will need the following fact: if the page $\sigma(t+1)$ is present in the cache of algorithm $\A$, i.e., $j \in \I_\A(t)$, where $\sigma(j)$ is the previous request for page $\sigma(t+1)$, and if in the graph $X(t)$ the vertex $j \in \I_\A(t)$ is connected by an edge to the vertex $i \in \I_\Q(t)$, then $i = j$. Indeed, by definition, $t+1 = \nu(j) \geq \nu(i)$. However, if $\nu(i) \leq t$, then index $i$ cannot belong to $\mathcal{I}_\mathcal{Q}(t)$, as it must have already been updated. Consequently, $\nu(i) = \nu(j) = t+1 \leq T$, which implies $i = j$.

Let $\Phi(t)$ denote the number of isolated vertices in the partition $\I_\Q(t)$ of the graph $X(t)$. We will show that for any $t$, the graphs $X(t)$ and $H_\Q(t)$ can be constructed in such a way that the following inequality holds:
\begin{equation}\label{Lemma1MainIneq}
\mathrm{OBJ}_\Q (t) + \Phi(t) \leq \mathrm{OBJ}_\A(t) + |E(t)|.
\end{equation}
Here, $H_\Q(t)$ is the eviction graph constructed after processing the first $t$ requests, and $|E(t)|$ is the number of edges in $H_\Q(t)$. The notation $\Delta F$ will denote $F(t+1) - F(t)$, where $F$ is any of the functions.

The proof proceeds by induction on $t$. For all $t$ corresponding to moments when the caches are not yet full, we define $X(t)$ as a perfect matching where the edges are of the form $(i, i)$. This is clearly possible since the cache contents are completely identical. The eviction graph will contain no edges. In this case, both sides of inequality \eqref{Lemma1MainIneq} are equal.

Now, assume the caches are full, request $\sigma(t)$ has been processed, and request $\sigma(t+1)$ has arrived. By the induction hypothesis, we already have a constructed matching $X(t)$ and a set of edges $E(t)$. Let us analyze the possible cases.

\begin{enumerate}
\item The page $\sigma(t+1)$ is present in both the cache of algorithm $\A$ and in the cache of algorithm $\Q$. This means that: 
\begin{itemize} 
\item $\I_\A(t+1) = \I_\A(t) \cup \{t+1\} \setminus \{j\}$,
\item $\I_\Q(t+1) = \I_\Q(t) \cup \{t+1\} \setminus \{j\}$,
\end{itemize}
where $\sigma(j)$ denotes the previous request for page $\sigma(t+1)$. The graph $X(t+1)$ is derived from $X(t)$ by:
\begin{itemize}
\item removing the vertex $j$ from both partitions, and
\item adding the vertex $t+1$, connected by the edge $(t+1, t+1)$.
\end{itemize}
Additionally:
\begin{itemize}
\item $\Delta \mathrm{OBJ}_\A = \Delta \mathrm{OBJ}_\Q = \Delta |E| = 0$,
\item $\Delta \Phi \leq 0$, because $j \in \I_\A(t)$ could have been connected by an edge only to $j \in \I_\Q(t)$ (as previously proven).
\end{itemize}
\item The page $\sigma(t+1)$ is present only in the cache of algorithm $\B$. In this case:
\begin{itemize}
\item $\I_\A(t+1) = \I_\A(t) \cup \{t+1\} \setminus \{a\}$,
\item $\I_\Q(t+1) = \I_\Q(t) \cup \{t+1\} \setminus \{j\}$,
\end{itemize} where $\sigma(a)$ is the page evicted by algorithm $\A$, i.e. $a = p_\A(t+1)$, and $\sigma(j)$ is the previous request for page $\sigma(t+1)$. The graph $X(t+1)$ is derived from $X(t)$ by:
\begin{itemize}
\item removing the vertices $a \in \I_\A(t)$ and $j \in \I_\Q(t)$, and
\item adding the vertex $t+1$, connected by the edge $(t+1, t+1)$.
\end{itemize} In this case:
\begin{itemize}
\item $\Delta \mathrm{OBJ}_\Q = \Delta |E| = 0$,
\item $\Delta \mathrm{OBJ}_\A = 1$,
\item $\Delta \Phi \leq 1$.
\end{itemize}

\item The page $\sigma(t+1)$ is present only in the cache of algorithm $\A$. In this case:
\begin{itemize}
\item $\I_\A(t+1) = \I_\A(t) \cup \{t+1\} \setminus \{j\}$,
\item $\I_\Q(t+1) = \I_\Q(t) \cup \{t+1\} \setminus \{b\}$,
\end{itemize} where $\sigma(b)$ is the page evicted by algorithm $\Q$, i.e. $b = p_\Q(t+1)$, and $\sigma(j)$ is the previous request for page $\sigma(t+1)$. As mentioned earlier, the vertex $j \in \I_\A(t)$ could have been connected by an edge only to $j$, and since $j \notin \I_\Q(t)$, it follows that the vertex $j \in \I_\A(t)$ is an isolated vertex in the graph $X(t)$. Furthermore, since the edges of the graph $X(t)$ do not intersect and the partitions are of equal size, the existence of an isolated vertex among the vertices of $\I_\A(t)$ implies the presence of an isolated vertex among the vertices of $\I_\Q(t)$.

We construct the graph $X(t+1)$ from $X(t)$ as follows:
\begin{itemize}
\item remove the vertices $j \in \I_\A(t)$ and $b \in \I_\Q(t)$, and
\item add the vertex $t+1$, connected by the edge $(t+1, t+1)$.
\end{itemize}
Notice that:
\begin{itemize}
\item $\Delta \mathrm{OBJ}_\A = 0$,
\item $\Delta \mathrm{OBJ}_\Q = 1$.
\end{itemize}
Now, let us consider the possible cases.
\begin{enumerate}
\item The vertex $b \in \I_\Q(t)$ was isolated in the graph $X(t)$. In this case, $\Delta \Phi = -1$, and inequality \eqref{Lemma1MainIneq} remains satisfied, since $\Delta |E| = 0$.
\item The vertex $b \in \I_\Q(t)$ was connected by an edge to some vertex $a \in \I_\A(t)$ in the graph $X(t)$, and among the isolated vertices in $\I_\Q(t)$, there was a vertex $b'$ such that $\nu(b) \geq \nu(b')$. Note that $a \neq j$, since the vertex $j$ is isolated, as previously mentioned. In this case, in the constructed graph $X(t+1)$, we can additionally add an edge ${a, b'}$, using the fact that $\nu(a) \geq \nu(b) \geq \nu(b')$, and the vertex $a \in \I_\A(t+1)$ becomes isolated because the vertex $b \in \I_\Q(t)$ was removed. After this, $\Delta \Phi = -1$, $\Delta |E| = 0$, and inequality \eqref{Lemma1MainIneq} is again satisfied.
\item \label{3c} The vertex $b \in \I_\Q(t)$ was connected by an edge to the vertex $a \in \I_\A(t)$ in the graph $X(t)$, and for any isolated vertex $\tilde{b} \in \I_\Q(t)$, it holds that $\nu(b) < \nu(\tilde{b})$. As we know, at least one such vertex $\tilde{b}$ exists, and therefore algorithm $\Q$ made a mistake by removing page $\sigma(b)$. In this case, we add the edge $(\tilde{b}, b)$ in the graph $H_\Q(t+1)$. After this, $\Delta \Phi = 0$, $\Delta |E| = 1$, and inequality \eqref{Lemma1MainIneq} is satisfied.
\end{enumerate}
\item The page $\sigma(t+1)$ is absent in the caches of both algorithms. In this case:
\begin{itemize}
\item $\I_\A(t+1) = \I_\A(t) \cup \{t+1\} \setminus \{a\}$,
\item $\I_\Q(t+1) = \I_\Q(t) \cup \{t+1\} \setminus \{b\}$,
\end{itemize}
where $\sigma(a)$ and $\sigma(b)$ are the pages evicted by algorithms $\A$ and $\Q$, respectively, i.e. $a = p_\A(t+1)$ and $b = p_\B(t+1)$. We construct the graph $X(t+1)$ from the graph $X(t)$ as follows:
\begin{itemize}
\item remove the vertices $a \in \I_\A(t)$ and $b \in \I_\Q(t)$,
\item add the vertex $t+1$ and connect it by the edge $(t+1, t+1)$.
\end{itemize}
Notice that $\Delta \mathrm{OBJ}_\A = \Delta \mathrm{OBJ}_\Q = 1$. Now, let us consider the possible cases.
\begin{enumerate}
\item At least one of the vertices $a \in \I_\A(t)$ or $b \in \I_\Q(t)$ was isolated in the graph $X(t)$, or the vertices $a$ and $b$ were connected by an edge. In this case, $\Delta \Phi \leq 0$, $\Delta |E| = 0$, and inequality \eqref{Lemma1MainIneq} is satisfied.
\item In the graph $X(t)$, there were edges $(a, b')$ and $(a', b)$, with $a \neq a'$ and $b \neq b'$. Notice that in the constructed graph $X(t+1)$, the vertex $b' \in \I_\Q(t+1)$ remains isolated, so there is at least one isolated vertex in $\I_\Q(t+1)$. If among the isolated vertices in the partition $\I_\Q(t+1)$, there is a vertex $\tilde{b}$ such that $\nu(b) \geq \nu(\tilde{b})$, we can additionally add the edge $(a', \tilde{b})$ in the graph $X(t+1)$. Indeed, $\nu(a') \geq \nu(b)$, since the edge $(a', b)$ was present in the graph $X(t)$, which, by the induction hypothesis, satisfies the invariants. Hence, $\nu(a') \geq \nu(\tilde{b})$. After this, $\Delta \Phi = 0$, $\Delta |E| = 0$, and inequality \eqref{Lemma1MainIneq} is satisfied.
\item \label{4c} If, for any isolated vertex $\tilde{b} \in \I_\Q(t+1)$, it holds that $\nu(b) < \nu(\tilde{b})$ (and such vertices, as we know, do exist), this means that algorithm $\Q$ made a mistake by removing page $\sigma(b)$. In this case, we choose any isolated vertex $\tilde{b} \in \I_\Q(t+1)$ and add the edge $(\tilde{b}, b)$ in the graph $H_\Q(t+1)$. In this case, $\Delta |E| = 1$, and inequality \eqref{Lemma1MainIneq} is satisfied, since $\Delta \Phi = 1$.
\end{enumerate}
\end{enumerate}
\end{proof}

\begin{corollary}\label{Cor_1}
    Let $\B$ be the \textsc{BlindOracle} algorithm. Then
    $$
    \mathrm{OBJ}_\B \leq \mathrm{OPT} + \eta.
    $$
\end{corollary}
\begin{proof}
Let $H_\B$ be the eviction graph from Lemma~\ref{Lemma_ErrorGraphExistence}. Consider any vertex $i$ in $H_\B$ with outgoing edges to $j_1, \ldots, j_{d_i}$, where $\nu(j_1) < \nu(j_2) < \cdots < \nu(j_{d_i}) < \nu(i)$ (see Lemma~\ref{Lemma_EdgeProperties} for details). Consequently, $\nu(i) - \nu(j_1) \geq d_i$. Applying Lemma~\ref{Lemma_LeftRightError}, we obtain
\[
\eta^-(i) + \eta^+(j_1) \geq \nu(i) - \nu(j) \geq d_i.
\]
Thus
\[
s(i) := \eta^-(i) + \sum_{r=1}^{d_i} \eta^+(j_r) \geq d_i.
\]
Summing over all vertices $i$ in $H_\B$ gives us
\[
\sum_{i=1}^T s(i) \geq \sum_{i=1}^T d_i = |E|,
\]
where $|E|$ is the total number of edges in $H_\B$. On the other hand, since $H_\B$ is an oriented forest (Lemma~\ref{Lemma_ErrorGraphForest}):
\[
\sum_{i=1}^T s(i) \leq \sum_{i=1}^T \left(\eta^-(i) + \eta^+(i)\right) = \sum_{i=1}^T \eta(i) = \eta.
\]
Therefore, $|E| \leq \eta$. Applying Lemma~\ref{Lemma_ErrorGraphExistence}, we conclude:
\[
\text{OBJ}_\B \leq \text{OPT} + |E| \leq \text{OPT} + \eta.
\]
\end{proof}

Corollary \ref{Cor_1} enables us to directly deduce that the competitive ratio of the \textsc{BlindOracle} algorithm is bounded above by $1 + \frac{\eta}{\mathrm{OPT}}$.

As we mentioned before, there might be many feasible eviction graphs. Moreover, even graph obtained in the proof of Lemma \ref{Lemma_ErrorGraphExistence} is defined ambiguously, since it leaves some freedom to choose the vertices from which the edges are drawn. Naturally, the question arises whether there are graphs satisfying any special properties. We will give an example of such a graph in the following lemma.

\begin{lemma} \label{Lemma_ColoredErrorGraphExistence1}
    For any algorithm $\mathcal{Q}$, there exists an eviction graph $H_\mathcal{Q}$ with colored edges that satisfies the following properties:
    \begin{itemize}
        \item If two edges $(i, j)$ and $(i', j')$ have the same color, then $i = i'$.
        \item If two edges $(i, j)$ and $(i', j')$ have the same color, then the evictions of pages $\sigma(j)$ and $\sigma(j')$ do not overlap.
        \item The total number of colors used is at most $\mathrm{OPT}$.
    \end{itemize}
    Furthermore, the following inequality holds:
    \[
    \mathrm{OBJ}_\mathcal{Q} \leq \mathrm{OPT} + |E|,
    \]
    where $E$ is the set of edges in the graph $H_\mathcal{Q}$.
\end{lemma}

\begin{proof}
    The construction of the graph $H_\Q$ follows the same procedure as in the proof of Lemma~\ref{Lemma_ErrorGraphExistence}, with one key difference: in steps \ref{3c} and \ref{4c}, instead of choosing an arbitrary isolated vertex $\tilde{b} \in \I_\Q(t)$ to draw an edge $(\tilde{b}, b)$, we will specify the selection more precisely. We will also describe how to assign colors to the edges to satisfy the required properties. Obviously, these modifications do not affect the inequality
    \[
    \mathrm{OBJ}_\mathcal{Q} \leq \mathrm{OPT} + |E|.
    \]

    Let us examine steps \ref{3c} and \ref{4c} from the proof of Lemma~\ref{Lemma_ErrorGraphExistence} more closely, as these are the only steps where an edge is added to $H_\Q$. In both cases, the vertex $b \in \I_Q(t)$, corresponding to the page $\sigma(b)$ evicted by algorithm $\Q$, is not an isolated vertex in the graph $X(t)$. It has a neighbor $a \in \I_\A(t)$, and this vertex $a$ becomes isolated in $X(t+1)$ after the removal of $b$. We denote this vertex by $h(b)$; it is the last neighbor of $b$ before its removal. Note that $h(b)$ is always defined for any vertex $b$ that has an incoming edge in $H_\Q$.

    We now introduce the concept of \emph{reservations}. The idea is simple: instead of selecting any arbitrary isolated vertex $\tilde{b}$ in steps \ref{3c} and \ref{4c}, we select an arbitrary isolated and unreserved vertex. To do this, we must define how vertices are reserved, how reservations are broken, and show that there is always at least one isolated and unreserved vertex available.

    The reservation procedure is as follows: whenever a vertex $\tilde{b}$ is chosen and an edge $(\tilde{b}, b)$ is added to $H_\Q$, the vertex $\tilde{b}$ is immediately marked as reserved. This is the only way a new reservation is created.

    To manage reservations, we introduce an auxiliary bipartite graph $\widetilde{X}(t)$, whose partitions correspond to those of $X(t)$. In this graph, a vertex $c \in \I_\Q(t)$ is considered reserved if and only if it is not isolated in $\widetilde{X}(t)$; thus, edges in $\widetilde{X}(t)$ represent reservations.

    We now describe how $\widetilde{X}(t+1)$ is constructed from $\widetilde{X}(t)$. First, the vertex set is updated to reflect changes in $\mathcal{I}_\mathcal{A}$ and $\mathcal{I}_\mathcal{Q}$. If a vertex is removed, all incident edges are also removed, which is the only way a reservation can be broken (since we do not delete edges for any other reason).

    Suppose now that $\widetilde{X}(t+1)$ is obtained from $\widetilde{X}(t)$ by updating its vertex set and deleting redundant edges. We impose two invariants on $\widetilde{X}(t)$:
    \begin{itemize}
        \item Its edges are disjoint (no two edges share a vertex).
        \item If $a \in \I_\mathcal{A}(t)$ has an edge in $\widetilde{X}(t)$, then $a$ is isolated in $X(t)$.
    \end{itemize}

    Removing edges cannot violate these invariants. To see this, observe that the proof of Lemma~\ref{Lemma_ErrorGraphExistence} implies the following: once a vertex $a \in \I_\mathcal{A}(t)$ becomes isolated in $X(t)$, it remains isolated until it is removed from $\I_\mathcal{A}$. Thus, our goal is to ensure that adding a new edge in $\widetilde{X}(t+1)$ does not violate these invariants. By the induction hypothesis, $\widetilde{X}(t)$ satisfies all the required constraints. A straightforward verification shows that $\widetilde{X}(t+1)$ also satisfies these constraints after updating the vertex set.

    As previously mentioned, the vertex $h(b) \in \I_\A(t+1)$ is isolated in $X(t+1)$, and $h(b) \in \I_\A(t)$ was not isolated in $X(t)$. Since $\widetilde{X}(t)$ satisfies the invariants, $h(b)$ is also isolated in $\widetilde{X}(t)$. Therefore, after updating the vertex set, $h(b)$ remains isolated in $\widetilde{X}(t+1)$. Recall that $\Phi(t+1)$ is the number of isolated vertices in the partition $\I_\Q(t+1)$ of $X(t+1)$. The number of isolated vertices in $\I_\A(t+1)$ is also $\Phi(t+1)$. Let $m$ be the number of edges in $\widetilde{X}(t+1)$. Since edges in $\widetilde{X}(t+1)$ are disjoint and can only be incident to isolated vertices of $\I_\A(t+1)$ in $X(t+1)$, we have $m \leq \Phi(t+1)$. Moreover, since $h(b)$ is isolated in both $X(t+1)$ and $\widetilde{X}(t+1)$, $m < \Phi(t+1)$. Thus, there must exist a vertex $\tilde{b} \in \I_Q(t+1)$ that is isolated in both $X(t+1)$ and $\widetilde{X}(t+1)$ --- that is, an unreserved isolated vertex. We reserve this vertex $\tilde{b}$ by adding the edge $(h(b), \tilde{b})$ to $\widetilde{X}(t+1)$, which does not violate any constraints.

    Next, we describe the coloring of edges in $H_\mathcal{Q}$. If $(\tilde{b}, b)$ is the first edge originating from $\tilde{b}$, we assign it a new color. Otherwise, we consider the edge $(\tilde{b}, b')$ --- the last edge previously added from $\tilde{b}$ --- and attempt to assign $(\tilde{b}, b)$ the same color. If this is not possible (i.e., if the evictions of $\sigma(b')$ and $\sigma(b)$ overlap), we assign a new color.

    To bound the number of colors, we construct an injection from the set of colors to the evictions performed by algorithm $\mathcal{A}$. A new color is introduced in two cases: when $(\tilde{b}, b)$ is the first outgoing edge from $\tilde{b}$, or when the evictions of $\sigma(b')$ and $\sigma(b)$ overlap, where $(\tilde{b}, b')$ is the previously added edge.

    First, consider the case where $(\tilde{b}, b)$ is the first outgoing edge. When this edge is drawn in $H_\mathcal{Q}$, the vertex $\tilde{b} \in \mathcal{I}_\mathcal{Q}$ is always isolated in $X$. From the proof of Lemma~\ref{Lemma_ErrorGraphExistence}, $X$ is constructed so that every new vertex $c$ is initially paired with an edge $(c, c)$. The edge $(\tilde{b}, \tilde{b})$ can only be removed if algorithm $\mathcal{A}$ evicts page $\sigma(\tilde{b})$. Thus, we map the color to this eviction.

    Next, consider the case where the evictions of $\sigma(b')$ and $\sigma(b)$ overlap, requiring a new color. Recall that $\sigma(b)$ is evicted at time $t + 1$. By definition of overlap, $\nu(b') > t + 1$. After adding edge $(\tilde{b}, b')$ to $H_\Q$, the vertex $\tilde{b}$ was reserved by adding the edge $(h(b'), \tilde{b})$ to $\widetilde{X}$. Since $\nu(h(b')) \geq \nu(b') > t + 1$ (by the constraints on $X$), the reservation was broken due to the removal of $h(b')$, which can only happen if algorithm $\mathcal{A}$ evicts page $\sigma(h(b'))$. We map the new color to this eviction.

    To ensure that this mapping is injective, we must check that no two colors correspond to the same eviction. The only nontrivial case is when one color corresponds to the first outgoing edge and the other to an overlap. However, in the first case, algorithm $\mathcal{A}$ evicts page $\sigma(\tilde{b})$, which is isolated in $\widetilde{X}$, while in the second case, the eviction breaks a reservation, i.e., algorithm $\mathcal{A}$ evicts page $\sigma(a')$ that has an edge in $\widetilde{X}$. Thus, the mapping is injective.
\end{proof}

\begin{corollary} \label{Cor_2}
    Let $\B$ denote the \textsc{BlindOracle} algorithm. Then, the following inequality holds:
    \[
    \mathrm{OBJ}_\B \leq 3 \cdot \mathrm{OPT} + \frac{3\eta}{k}.
    \]
\end{corollary}

\begin{proof}
    Consider the eviction graph $H_\B$ from Lemma~\ref{Lemma_ColoredErrorGraphExistence1}. Focus on a vertex $x$ with outgoing edges grouped by colors: edges $(x, y^1_1), \ldots, (x, y^1_{d_1})$ of the first color, $(x, y^2_1), \ldots, (x, y^2_{d_2})$ of the second color, and so on. For each color $r$, we order the edges so that $\nu(y^r_1) < \nu(y^r_2) < \ldots < \nu(y^r_{d_r})$. Let $t^r_i$ be the time when page $\sigma(y^r_i)$ was evicted. Since the evictions do not overlap, we have the following chain of inequalities:
    \begin{equation} \label{Cor_2_ineq_1}
        t^r_1 < \nu(y^r_1) \leq t^r_2 < \nu(y^r_2) \leq \ldots \leq t^r_{d_r} < \nu(y^r_{d_r}).
    \end{equation}
    Let $n_x$ be the total number of colors outgoing from $x$. Without loss of generality, assume that the colors are ordered according to the order in which they were added during the construction of $H_\B$ (as described in the proof of Lemma~\ref{Lemma_ColoredErrorGraphExistence1}). Thus,
    \begin{equation} \label{Cor_2_ineq_2}
    t^1_1 < \ldots < t^1_{d_1} < t^2_1 < \ldots < t^2_{d_2} < \ldots < t^{n_x}_1 < \ldots < t^{n_x}_{d_{n_x}}.
    \end{equation}
    Combining inequalities \eqref{Cor_2_ineq_1} and \eqref{Cor_2_ineq_2}, and using the fact that $\nu(x) > \nu(y^{n_x}_{d_{n_x}})$, we obtain:
    \begin{equation} \label{Cor_2_ineq_3}
        \nu(x) - \nu(y^1_1) \geq \sum_{r=1}^{n_x} \sum_{i=2}^{d_r-1}  L^r_i,
    \end{equation}
    where $L^r_i$ is the distance between the eviction of page $\sigma(y^r_i)$ and its next request, i.e., $L^r_i := \nu(y^r_i) - t^r_i$.

    Applying Lemma~\ref{Lemma_LeftRightError} to inequality \eqref{Cor_2_ineq_3}, we get
    \begin{equation} \label{Cor_2_ineq_4}
        \eta^-(x) + \eta^+(y^1_1) \geq \sum_{r=1}^{n_x} \sum_{i=2}^{d_r-1} L_i^r.
    \end{equation}

    For any $y=1,\ldots,T$, let $M(y)$ denote the number of pages $y'$ such that $\nu(y') > \nu(y)$ but $\omega(y') \leq \omega(y)$; these are called \textit{inversions}. Adding $\sum_{r=1}^{n_x} \sum_{i=2}^{d_r-1} M(y^r_i)$ to both sides of inequality \eqref{Cor_2_ineq_4}, we obtain
    \begin{equation} \label{Cor_2_ineq_5}
        \eta^-(x) + \eta^+(y^1_1) + \sum_{r=1}^{n_x} \sum_{i=2}^{d_r-1} M(y^r_i) \geq \sum_{r=1}^{n_x} \sum_{i=2}^{d_r-1} (L_i^r + M(y_i^r)).
    \end{equation}
    The reason for this is as follows: for every $y_i^r$, the following inequality holds:
    \begin{equation} \label{Cor_2_ineq_6}
        L_i^r + M(y_i^r) \geq k.
    \end{equation}
    Indeed, if $L_i^r \geq k$, then \eqref{Cor_2_ineq_6} is obvious. If $L_i^r < k$, note that there are at most $L^r_i - 1$ pages whose next request occurs in the interval $(t_i^r, \nu(y_i^r))$. Therefore, at least $k - L_i^r$ pages that were in the cache together with $\sigma(y_i^r)$ at the time of its eviction have $\nu$-values strictly greater than $\nu(y_i^r)$. Since \textsc{BlindOracle} chose to evict $\sigma(y_i^r)$ instead of these pages, their $\omega$-values do not exceed $\omega(y_i^r)$, so $M(y_i^r) \geq k - L_i^r$. Thus, inequality \eqref{Cor_2_ineq_6} holds.

    Substituting \eqref{Cor_2_ineq_6} into \eqref{Cor_2_ineq_5}, we get
    \[
        \eta^-(x) + \eta^+(y^1_1) + \sum_{r=1}^{n_x} \sum_{i=2}^{d_r-1} M(y^r_i) \geq k(|E_x| - 2),
    \]
    where $E_x$ denotes the set of edges outgoing from vertex $x$.
    
    Summing over all vertices $x$ and using the fact that $H_\B$ is an oriented forest (by Lemma~\ref{Lemma_ErrorGraphForest}), we obtain
    \[
    \eta + M \geq k (|E| - 2n),
    \]
    where $E$ is the set of edges in $H_\B$, $n$ is the total number of colors, and $M$ is the total number of inversions. It is known from~\cite{PredictiveMarkerImprovement1} that $M \leq 2\eta$. Also, from Lemma~\ref{Lemma_ColoredErrorGraphExistence1}, we know that $n \leq \mathrm{OPT}$, so
    \[
    |E| \leq \frac{3\eta}{k} + 2 \cdot \mathrm{OPT}.
    \]

    Again, from Lemma~\ref{Lemma_ColoredErrorGraphExistence1}, we have
    \[
    \mathrm{OBJ}_\B \leq \mathrm{OPT} + |E|,
    \]
    and combining these two inequalities gives the desired result:
    \[
    \mathrm{OBJ}_\B \leq 3 \cdot \mathrm{OPT} + \frac{3\eta}{k}.
    \]
\end{proof}

Corollary \ref{Cor_2} directly implies that the competitive ratio of the \textsc{BlindOracle} algorithm is upper-bounded by $3 + \frac{3}{k}\frac{\eta}{\mathrm{OPT}}$. Combining this result with Corollary \ref{Cor_1}, we conclude that the competitive ratio is bounded by
$$
    \min\left\{1 + \frac{\eta}{\mathrm{OPT}}, 3 + \frac{3}{k}\frac{\eta}{\mathrm{OPT}}\right\}.
$$

As previously discussed, a chain can be considered a particular case of non-overlapping evictions. Therefore, we now impose stricter conditions on the eviction graph, requiring that all edges of the same color must form a chain:

\begin{lemma} \label{Lemma_ColoredErrorGraphExistence}
For any algorithm $\Q$, there exists an eviction graph $H_\Q$ with colored edges that satisfies the following properties:
\begin{itemize}
    \item If two edges $(i, j)$ and $(i', j')$ have the same color, then $i = i'$.
    \item All edges of the same color can be ordered as $(i, j_1), \ldots, (i, j_d)$ so that the sequence $(j_1, \ldots, j_d)$ forms a chain.
    \item The total number of colors is at most $3 \cdot \mathrm{OPT}$.
\end{itemize}
Moreover, the following inequality holds:
\[
\mathrm{OBJ}_\Q \leq 2 \cdot \mathrm{OPT} + |E|,
\]
where $E$ is the set of edges in the graph $H_\Q$.
\end{lemma}

\begin{proof}
The proof relies heavily on the proof of Lemma \ref{Lemma_ColoredErrorGraphExistence1}, with some modifications. As before, we focus on steps \ref{3c} and \ref{4c}, and we continue to use the auxiliary reservation graph $\widetilde{X}$, but with some differences described below. It remains true that there always exists a vertex that is isolated in both $X(t+1)$ and $\widetilde{X}(t+1)$; we denote this vertex by $x$.

Suppose the eviction of $\sigma(b)$ was triggered by the eviction of $\sigma(b')$, and there is an edge $(c, b')$ in $H_\Q$ and an edge $(h(b'), c)$ in $\widetilde{X}(t)$ (where $h(b')$ is defined as in Lemma~\ref{Lemma_ColoredErrorGraphExistence1}). In this case, we extend the reservation of $c$ by replacing the edge $(h(b'), c)$ with $(h(b), c)$ in $\widetilde{X}(t+1)$. In all other cases, we reserve the vertex $x$ by adding the edge $(h(b), x)$ to $\widetilde{X}(t+1)$.

Now, we describe how edges are added to $H_\Q$. If $x$ was chosen for reservation, we add an edge $(x, b)$ to $H_\Q$ with a new color. If the reservation of $c$ was extended, we add an edge $(c, b)$ to $H_\Q$ with the same color as $(c, b')$, unless $\sigma(c)$ is not better than $\sigma(b)$. In this latter case, we do not add an edge to $H_\Q$ (to avoid violating the definition of the eviction graph), but instead increment an auxiliary counter $\psi$, which counts all such situations where an edge should have been added but was not.

Obviously,
\[
\mathrm{OBJ}_\Q \leq \mathrm{OPT} + |E| + \psi,
\]
since there are exactly $\psi$ instances where we did not add an edge to $H_\Q$ even though we were supposed to.

Let us analyze all cases where a new color is used (i.e., when vertex $x$ was chosen for reservation):

\begin{itemize}
    \item The eviction of $\sigma(b)$ was not triggered by any other eviction. This means that the requested page $\sigma(t+1)$ is new, so any algorithm, including the optimal one, will incur a cache miss. The number of such cases is at most $\mathrm{OPT}$.
    \item The eviction of $\sigma(b)$ was triggered by the eviction of $\sigma(b')$, but the eviction of $\sigma(b')$ did not result in a reservation. This means that the eviction of $\sigma(b')$ did not correspond to cases \ref{3c} or \ref{4c}. We claim that this can happen at most $\mathrm{OPT}$ times. Indeed, these are cases where $\mathrm{OBJ}_\Q$ increases without an increase in $|E|$ or $\psi$. However, we know that $\mathrm{OBJ}_\Q \leq \mathrm{OPT} + |E| + \psi$ always holds.
    \item The eviction of $\sigma(b)$ was triggered by the eviction of $\sigma(b')$, and the vertex $c$ was reserved by $b'$. This means that the edge $(h(b'), c)$ was created in $\widetilde{X}$, but is now absent in $\widetilde{X}(t)$. We claim that the number of such edges is at most $\mathrm{OPT}$. The proof is similar to the argument for bounding the number of colors in Lemma~\ref{Lemma_ColoredErrorGraphExistence1}, so we only sketch it here. We construct an injection from such edges to evictions made by algorithm~$\mathcal{A}$. The edge $(h(b'), c)$ can only be removed if one of its endpoints is deleted. The vertex $h(b')$ can only be deleted if algorithm~$\mathcal{A}$ evicts the page $\sigma(h(b'))$, since $\nu(h(b')) \geq \nu(b') = t + 1$ (because the eviction of $\sigma(b')$ triggered the eviction of $\sigma(b)$ at time $t + 1$). Thus, we map this edge to that eviction. The removal of the other endpoint, $c$, can be mapped to the eviction of its twin neighbor, i.e., the vertex $c$ from the other partition whose eviction destroyed the initial edge $(c, c)$.

\end{itemize}

Summing up, we see that the total number of colors does not exceed $3 \cdot \mathrm{OPT}$. It remains to bound the counter $\psi$.

Each increment of $\psi$ is associated with the creation of an edge $(h(b), c)$ in $\widetilde{X}(t+1)$. This edge is either destroyed when one of its endpoints is removed, or it persists until the end. As shown above, the number of such edges is at most $\mathrm{OPT}$.

\end{proof}

\section{Alternating Strategies} \label{Sect:UpperBound}

One of the main advantages of using chains is that they establish a clear connection between the algorithm and its eviction graph, as formalized in Lemma~\ref{Lemma_ColoredErrorGraphExistence}. Importantly, an online algorithm cannot construct the eviction graph during its execution, since it only becomes aware of its mistakes after they occur. However, it can maintain chains, as it always knows which previous eviction led to the current cache miss. By ensuring that the algorithm follows a specific behavior along each chain, and utilizing the fact that edges of the same color form a chain, we can leverage this structure for further analysis.

In our approach, we specifically design the algorithm to alternate between different strategies along the chain. More precisely, when we say that algorithm $\mathcal{Q}$ alternates strategies $\mathcal{Q}_1, \ldots, \mathcal{Q}_s$, we mean the following: if the eviction of page $p_\mathcal{Q}(t)$ was performed using strategy $\mathcal{Q}_i$ and this eviction subsequently triggered the eviction of page $p_\mathcal{Q}(t')$, then $p_\mathcal{Q}(t')$ should be evicted using strategy $\mathcal{Q}_{i \oplus 1}$, where $i \oplus 1$ denotes $(i + 1 \bmod s) + 1$.

Our algorithm, as described in Section~\ref{Sec:Alg}, alternates between three strategies: \textsc{BlindOracle}, \textsc{RandomAlg}, and \textsc{Corrector}.

\begin{figure}[h]
    \centering
    \includegraphics[width=0.95\textwidth]{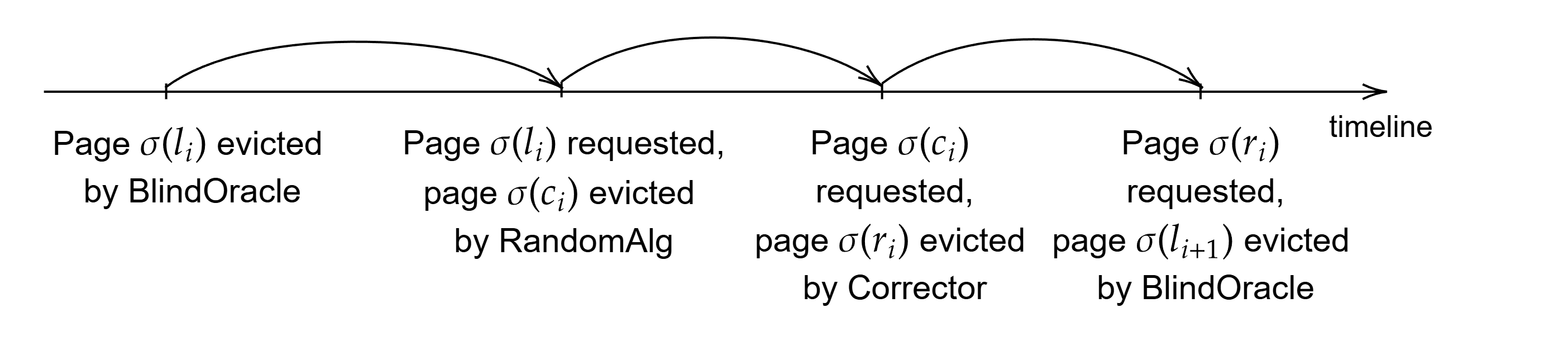}
    \caption{Triple $(l_i, c_i, r_i)$.}
    \label{fig:TripleExample}
\end{figure}

\begin{figure} [h]
    \centering
    \includegraphics[width=0.95\textwidth]{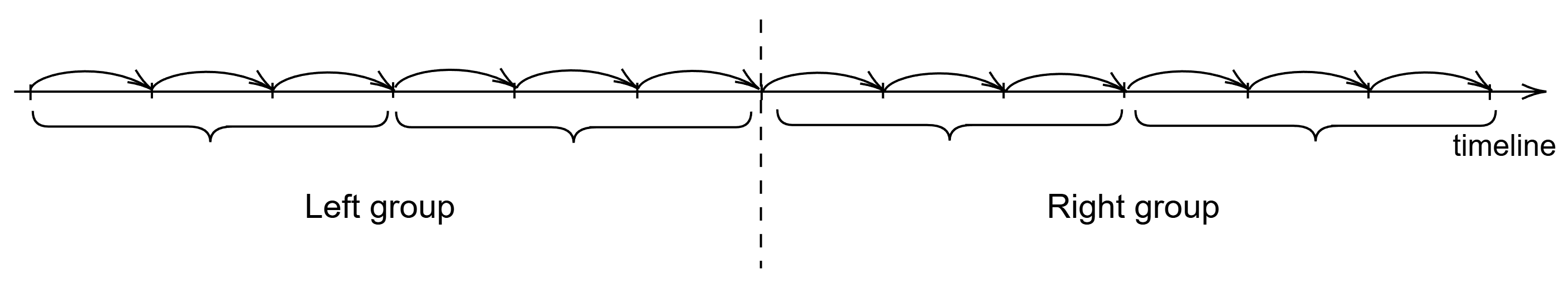}
    \caption{An example of a divisible chain, consisting of four triples.}
    \label{fig:GroupsExample}
\end{figure}

For our analysis, it is convenient to consider triples of consecutive links in the chain. The $i$-th triple corresponds to the tuple $(l_i, c_i, r_i)$, where page $\sigma(l_i)$ was evicted by \textsc{BlindOracle}, $\sigma(c_i)$ was chosen randomly, and $\sigma(r_i)$ was evicted by \textsc{Corrector} (see Fig.~\ref{fig:TripleExample}). We are particularly interested in $L_i$, the length of the intermediate link. Specifically, $L_i$ is defined as $t' - t$, where $t$ is the time step when $\sigma(c_i)$ was evicted, and $t'$ is the time step when $\sigma(r_i)$ was evicted. We refer to a triple as \textit{short} if $L_i \leq k/10$.

Consider an arbitrary chain $(b_1, \ldots, b_d)$. We call it \textit{divisible} if its links can be partitioned into an even number of triples. Note that the triples of a divisible chain are naturally divided into two groups, the left group and the right group, each containing an equal number of triples (see Fig.~\ref{fig:GroupsExample}).

\begin{lemma} \label{Lemma_LongChain}
    Let $\mathcal{B}$ be the \textsc{AlternatingOracle}, $H_\mathcal{B}$ the eviction graph from Lemma~\ref{Lemma_ColoredErrorGraphExistence}, and let the edges $(a, b_1), \ldots, (a, b_d)$ be of the same color, forming a divisible chain $(b_1, \ldots, b_d)$. Let $\psi$ denote the maximum number of short triples in either the left or right group of the chain. Then, the following inequality holds:
    $$
    \sum_{i=1}^d \eta(b_i) \geq \frac{(d/6 - \psi)^2 k}{20}.
    $$
\end{lemma}

\begin{proof}
    We consider two cases based on the position of $\omega(a)$: either $\omega(a)$ lies (not strictly) to the right of the left group of triples, or $\omega(a)$ lies strictly to the left of the right group. We begin with the first case. Let $h := d/3$ denote the total number of triples in the chain.

    Consider the first triple $(l_1, c_1, r_1)$. Algorithm $\mathcal{B}$ evicted $\sigma(l_1)$ instead of $\sigma(a)$ using the \textsc{BlindOracle} strategy, which implies $\omega(l_1) \geq \omega(a)$. Additionally, there are $h/2$ central links in the left group, corresponding to the eviction of random pages, located within the interval $(\nu(l_1), \omega(a))$. Therefore, $\eta(l_1) \geq \sum_{i=1}^{h/2} L_i$, where $L_i$ is the length of the central link in the $i$-th triple (see Fig.~\ref{fig:ErrorExample}).

    \begin{figure}[h]
        \centering
        \includegraphics[width=1\textwidth]{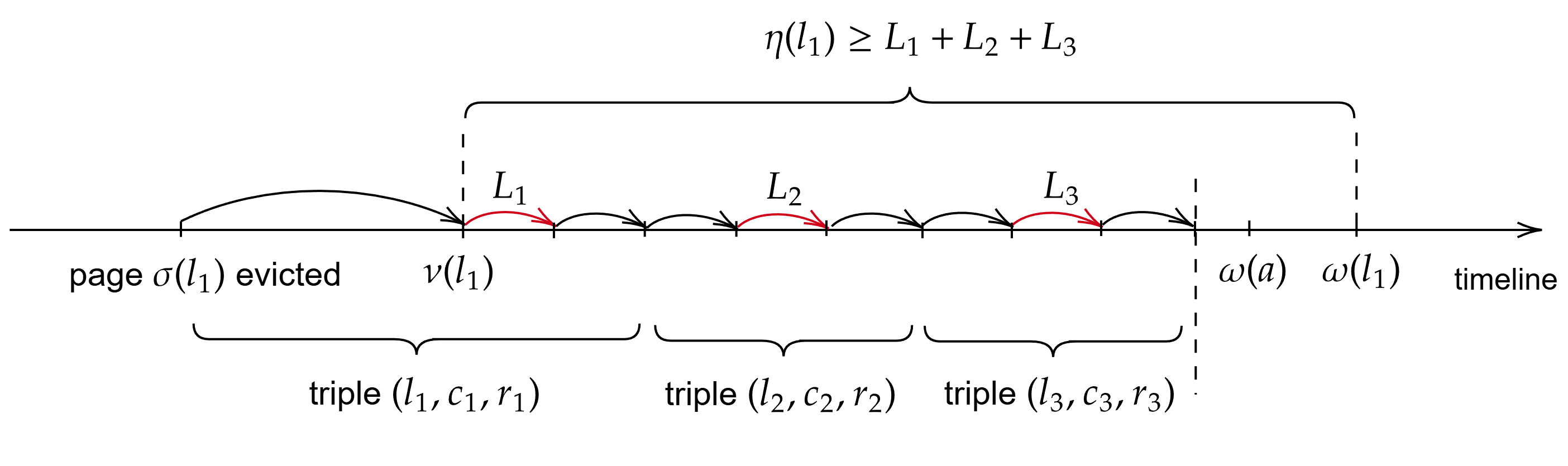}
        \caption{Illustration of how the loss $\eta(l_1)$ is estimated.}
        \label{fig:ErrorExample}
    \end{figure}

    Similarly, $\eta(l_2) \geq \sum_{i=2}^{h/2} L_i$, and so on. Summing these inequalities, we obtain:
    $$
    \sum_{i=1}^d \eta(b_i) \geq \sum_{i=1}^{h/2} \eta(l_i) \geq \sum_{i=1}^{h/2} i L_i \geq \sum_{i=1}^{h/2 - \psi} i \cdot \frac{k}{10} \geq \frac{(d/6 - \psi)^2 k}{20}.
    $$

    Now, assume the second case holds, i.e., $\omega(a)$ lies strictly to the left of the right group. Consider the eviction of $\sigma(r_1')$, where $r_1' := r_h$. The \textsc{Corrector} chose to evict $\sigma(r_1')$ instead of $\sigma(a)$, even though the incorrectness of the prediction $\omega(a)$ was already evident. Thus, $\omega(r_1') \leq \omega(a)$. Moreover, there are $h/2$ central links in the right group within the interval $(\omega(a), \nu(r_1'))$. Therefore, $\eta(r_1') \geq \sum_{i=1}^{h/2} L_i'$, where $L_i' := L_{h+1-i}$.

    By extending this argument, we obtain:
    $$
    \sum_{i=1}^d \eta(b_i) \geq \sum_{i=1}^{h/2} \eta(r_i') \geq \sum_{i=1}^{h/2} i L_i' \geq \sum_{i=1}^{h/2 - \psi} i \cdot \frac{k}{10} \geq \frac{(d/6 - \psi)^2 k}{20}.
    $$
\end{proof}

The following auxiliary result will be required for our analysis:

\begin{lemma} \label{Lemma_RandAlg}
    Let $\xi = (\xi_1, \xi_2, \ldots, \xi_T)$ be a tuple of independent and identically distributed (i.i.d.) random variables following a Bernoulli distribution with parameter $\gamma \in (0, 1)$. Suppose there exists a deterministic algorithm that takes as input the realization of $\xi = (\xi_1, \xi_2, \ldots, \xi_T)$ and outputs a natural number $N(\xi)$. Define $S(\xi, m)$ as the number of ones among the first $m$ values, i.e.,
    $$
    S(\xi, m) := \sum_{i=1}^{m} \xi_i.
    $$
    Then, as $T \to +\infty$ and $\E [N(\xi)] \to +\infty$, the ratio $\E [S(\xi, N(\xi))] \div \E [N(\xi)]$ converges to $\gamma$.
\end{lemma}

\begin{proof}
    Let $\alpha \in (\gamma, 1)$. We define two events: $A_\alpha := \{S(\xi, N(\xi)) \geq \alpha \cdot N(\xi)\}$ and its complement $\bar{A}_\alpha$. Additionally, let $N^*(\xi)$ denote the maximal $n^*$ such that $S(\xi, n^*) \geq \alpha n^*$. 

    Observe that 
    \begin{equation} \label{eq:lemma_randalg_1}
        \E[N(\xi) \mid A_\alpha] \mathbb{P}(A_\alpha) \leq \E [N^*(\xi) \mid A_\alpha] \mathbb{P}(A_\alpha) \leq \E N^*(\xi).
    \end{equation}
    Furthermore,
    $$
        \E [N^*(\xi)] \leq \sum_{n=1}^{\infty} n \mathbb{P} \{S(\xi, n) \geq \alpha n\}.
    $$
    Applying the Chernoff bound, we obtain
    $$
        \mathbb{P} \left\{\sum_{i=1}^n \xi_i \geq \alpha n\right\} \leq \exp(-\beta n),
    $$
    where $\beta := \delta^2 \gamma / (2 + \delta)$ and $\delta := \alpha/\gamma - 1$. Substituting this into the previous inequality yields
    \begin{equation} \label{eq:lemma_randalg_2}
        \E [N^*(\xi)] \leq \sum_{n=1}^{\infty} n \exp(-\beta n) = \frac{\exp(\beta)}{(\exp(\beta) - 1)^2}.
    \end{equation}
    Let $\theta := \exp(\beta)/(\exp(\beta) - 1)^2$. Note that $\theta$ is a constant depending only on $\alpha$ and $\gamma$. Combining \eqref{eq:lemma_randalg_1} and \eqref{eq:lemma_randalg_2}, we conclude
    \begin{equation} \label{eq:lemma_randalg_3}
        \E[N(\xi) \mid A_\alpha] \mathbb{P}(A_\alpha) \leq \theta.
    \end{equation}

    Next, we decompose $\E [S(\xi, N(\xi))]$ as follows:
    $$
    \E [S(\xi, N(\xi))] = \E [S(\xi, N(\xi)) \mid A_\alpha] \mathbb{P}(A_\alpha) + \E [S(\xi, N(\xi)) \mid \bar{A}_\alpha] \mathbb{P}(\bar{A}_\alpha).
    $$
    Using the inequalities $\E [S(\xi, N(\xi)) \mid A_\alpha] \leq \E [N(\xi) \mid A_\alpha]$ and $\E [S(\xi, N(\xi)) \mid \bar{A}_\alpha] \leq \alpha \E [N(\xi) \mid \bar{A}_\alpha]$, we obtain
    $$
    \E [S(\xi, N(\xi))] \leq \E [N(\xi) \mid A_\alpha] \mathbb{P}(A_\alpha) + \alpha \E [N(\xi) \mid \bar{A}_\alpha] \mathbb{P}(\bar{A}_\alpha).
    $$
    Combining this with \eqref{eq:lemma_randalg_3}, we derive
    $$
    \E [S(\xi, N(\xi))] \leq \theta + \alpha \E [N(\xi)].
    $$
    As $\E [N(\xi)] \to \infty$, the ratio $\E [S(\xi, N(\xi))] \div \E [N(\xi)]$ is asymptotically bounded above by $\gamma$, since $\alpha$ can be chosen arbitrarily close to $\gamma$.

    To establish the lower bound $\E [S(\xi, N(\xi))] \div \E [N(\xi)] \geq \gamma$, it suffices to show that the ratio $\E [\bar{S}(\xi, N(\xi))] \div \E [N(\xi)]$ is asymptotically bounded above by $1 - \gamma$, where $\bar{S}(\xi, m) := m - S(\xi, m)$ represents the number of zeros among the first $m$ values. This follows directly from the fact that each zero occurs with probability $1 - \gamma$, and the previously proven result can be applied.
\end{proof}

\begin{lemma} \label{Lemma_MainResult}
    Suppose $k \geq 10$ and $\B$ is the \textsc{AlternatingOracle}. Then the following inequality holds:
    $$
    \mathbb{E}[\mathrm{OBJ}_\mathcal{B}] \leq \max \left\{230 \cdot \mathrm{OPT} + 60 \sqrt{\frac{ 60  \cdot \eta \cdot \mathrm{OPT}}{k}},\, c\right\},
    $$
    where $c$ is an absolute constant.
\end{lemma}

\begin{proof}
    Let $H_\mathcal{B}$ denote the eviction graph from Lemma~\ref{Lemma_ColoredErrorGraphExistence}. As previously established, for any realization of the random bits, the following holds:
    \begin{equation} \label{Lemma_MainResult_ineq_1}
        \mathrm{OBJ}_\mathcal{B} \leq 2 \cdot \mathrm{OPT} + |E|.
    \end{equation}

    First, we ensure that all chains corresponding to colors are divisible. To achieve this, we remove up to four edges at the ends of each chain to allow partitioning into triples, and, if necessary, delete up to three additional edges to ensure the number of triples is even (or to eliminate the color entirely, since zero is also even). In total, we remove at most seven edges per color, which reduces $|E|$ by at most $21 \cdot \mathrm{OPT}$, since there are at most $3 \cdot \mathrm{OPT}$ colors. Thus, from~\eqref{Lemma_MainResult_ineq_1}, we obtain:
    \begin{equation} \label{Lemma_MainResult_ineq_2}
        \mathrm{OBJ}_\mathcal{B} \leq 23 \cdot \mathrm{OPT} + |E'|,
    \end{equation}
    where $E'$ is the set of edges after making all chains divisible.

    Let $n$ denote the number of remaining colors, and let $d_i$ be the number of edges of the $i$-th color. Denote these edges as $(a, b_1), \ldots, (a, b_{d_i})$, and define
    $$
    s(i) := \sum_{x=1}^{d_i} \eta(b_{x}).
    $$
    Summing $s(i)$ over all colors, each term appears at most once, since $H_\mathcal{Q}$ is an oriented forest (see Lemma~\ref{Lemma_ErrorGraphForest}). Therefore, $\eta \geq \sum_{i=1}^n s(i)$. By Lemma~\ref{Lemma_LongChain}, for any color $i$,
    $$
    s(i) \geq \frac{(d_i/6 - \psi_i)^2 k}{20},
    $$
    where $\psi_i$ is the maximum number of short triples in either the left or right group of the chain.

    Consequently,
    $$
    \eta \geq \sum_{i=1}^n \frac{(d_i/6 - \psi_i)^2 k}{20}.
    $$
    Solving the corresponding optimization problem yields the upper bound
    $$
    \sum_{i=1}^n (d_i/6 - \psi_i) \leq \sqrt{\frac{20\eta n}{k}}.
    $$
    Since $n \leq 3 \cdot \mathrm{OPT}$, it follows that
    $$
    |E'| = \sum_{i=1}^n d_i \leq 6 \sqrt{\frac{60\eta \cdot \mathrm{OPT}}{k}} + 6 \sum_{i=1}^n \psi_i.
    $$
    Combining this result with~\eqref{Lemma_MainResult_ineq_2}, we obtain
    \begin{equation} \label{eq:lemma_final_result_1}
        \mathrm{OBJ}_\mathcal{B} \leq 23 \cdot \mathrm{OPT} + 6 \sqrt{\frac{60\eta \cdot \mathrm{OPT}}{k}} + 6 \sum_{i=1}^n \psi_i.
    \end{equation}
    This inequality holds for any realization of the random bits in algorithm $\mathcal{B}$. However, for certain realizations, the term $6 \sum_{i=1}^n \psi_i$ may dominate $\mathrm{OBJ}_\mathcal{B}$ (for example, if all chains are short), making~\eqref{eq:lemma_final_result_1} trivial. This is where the role of randomness becomes essential. We show that for sufficiently large $\mathbb{E}[\mathrm{OBJ}_\mathcal{B}]$, the following holds:
    $$
    \mathbb{E}\left[6 \sum_{i=1}^n \psi_i\right] \leq 0.9 \cdot \mathbb{E}[\mathrm{OBJ}_\mathcal{B}].
    $$

    To prove this, consider an arbitrary triple. The value $L_j$ equals $\nu(c_j) - t_j$, where $t_j$ is the time step when $\sigma(c_j)$ was evicted.

    Consider the cache content immediately before the eviction of $\sigma(c_j)$. Sort all stored pages by the remaining time until their next request. Let $q_j$ denote the position of the randomly chosen page in this sorted list (for example, $q_j = k$ if the algorithm accidentally selects the optimal page). We claim that $L_j \geq q_j$.

    Indeed, there are $q_j - 1$ pages in the cache that will be requested before the evicted page. Even if only one request occurs for each of these pages and no other requests are made, there will still be $q_j - 1$ requests before $\nu(c_j)$.

    It follows that $\sum_{i=1}^n \psi_i \leq \widetilde{\psi}$, where $\widetilde{\psi}$ counts the number of $q_i$ such that $q_i \leq k/10$. Since $q_i \sim U\{1, \ldots, k\}$ and $k \geq 10$, the probability $\gamma$ that $q_i \leq k/10$ is at most $0.1$.

    During its execution, algorithm $\mathcal{B}$ reads $N$ realizations of the random variables $q_i$, where $N$ depends on these realizations. For simplicity, assume $N = \mathrm{OBJ}_\mathcal{B}$, meaning the algorithm reads $q_i$ even when using \textsc{BlindOracle} or \textsc{Corrector} but does not use it. Since no other randomness is involved, $\mathcal{B}$ can be viewed as a deterministic algorithm with these random inputs. By applying Lemma~\ref{Lemma_RandAlg}, if $\mathbb{E}[\mathrm{OBJ}_\mathcal{B}] \geq c$ for some sufficiently large constant $c$, then $\mathbb{E}[\widetilde{\psi}] \leq 0.15 \cdot \mathbb{E}[\mathrm{OBJ}_\mathcal{B}]$. Therefore,
    $$
    6 \cdot \mathbb{E}\left[\sum_{i=1}^n \psi_i\right] \leq 6 \cdot \mathbb{E} [\widetilde{\psi}] \leq 0.9 \cdot \mathbb{E}[\mathrm{OBJ}_\mathcal{B}],
    $$
    as required. Substituting this into~\eqref{eq:lemma_final_result_1}, we obtain
    $$
    \mathbb{E}[\mathrm{OBJ}_\mathcal{B}] \leq \max \left\{23 \cdot \mathrm{OPT} + 6 \sqrt{\frac{60\eta \cdot \mathrm{OPT}}{k}} + 0.9 \cdot \mathbb{E}[\mathrm{OBJ}_\mathcal{B}],\, c\right\}.
    $$
    Rearranging terms completes the proof.
\end{proof}

It is important to note that although Lemma~\ref{Lemma_MainResult} was proven under the assumption that $k \geq 10$, this restriction does not affect the asymptotic analysis. Indeed, for constant values of $k$ (specifically, when $k < 10$), even the LRU algorithm achieves a constant competitive ratio.

Nevertheless, one might be concerned about the large constants appearing in our upper bound for the competitive ratio. For practical values of $\eta$ and $k$, these constants may seem unacceptably high. To address this concern, we show that, in general, the \textsc{AlternatingOracle} algorithm performs comparably to the \textsc{BlindOracle} algorithm.

\begin{lemma} \label{Lemma_NotMuchWorse}
    Let $\mathcal{B}$ denote the \textsc{AlternatingOracle} algorithm. Then the following inequality holds:
    $$
    \mathrm{OBJ}_\mathcal{B} \leq 3 \cdot \mathrm{OPT} + 3\eta.
    $$
\end{lemma}

\begin{proof}
    As a direct consequence of Lemma~\ref{Lemma_ErrorGraphExistence}, for any algorithm $\mathcal{Q}$, we have:
    \begin{equation} \label{eq:lemma_notmuchworse_1}
        \mathrm{OBJ}_{\mathcal{Q}} \leq \mathrm{OPT} + \mathrm{OBJ}^{-}_{\mathcal{Q}},
    \end{equation}
    where $\mathrm{OBJ}^{-}_{\mathcal{Q}}$ denotes the number of mistakes made by $\mathcal{Q}$, and $\mathrm{OBJ}^{+}_{\mathcal{Q}}$ denotes the number of times $\mathcal{Q}$ evicts the optimal page. Thus, $\mathrm{OBJ}_{\mathcal{Q}} = \mathrm{OBJ}^{+}_{\mathcal{Q}} + \mathrm{OBJ}^{-}_{\mathcal{Q}}$. Substituting this into~\eqref{eq:lemma_notmuchworse_1}, we obtain:
    \begin{equation} \label{eq:lemma_notmuchworse_2}
        \mathrm{OBJ}^{+}_{\mathcal{Q}} \leq \mathrm{OPT}.
    \end{equation}
    
    Let $H_{\mathcal{B}}$ be the eviction graph as defined in Lemma~\ref{Lemma_ErrorGraphExistence}. We have:
    \begin{equation} \label{eq:lemma_notmuchworse_3}
        \mathrm{OBJ}_{\mathcal{B}} \leq \mathrm{OPT} + |E|,
    \end{equation}
    where $E$ is the set of edges in $H_{\mathcal{B}}$. Although Lemma~\ref{Lemma_ErrorGraphExistence} does not require an edge to be added for every mistake, we can always augment $H_{\mathcal{B}}$ with additional edges without violating~\eqref{eq:lemma_notmuchworse_3}. Therefore, we may assume that $E$ contains an edge for every mistake.

    Let $B \subseteq E$ be the subset of edges corresponding to evictions made by the \textsc{BlindOracle} strategy. From the proof of Corollary~\ref{Cor_1}, we have $|B| \leq \eta$. Furthermore, we claim that $|E| \leq 3|B| + 2 \cdot \mathrm{OPT}$, which will suffice to complete the proof.

    We partition the evictions into triples, where each triple consists of:
    \begin{enumerate}
        \item An eviction made by \textsc{BlindOracle},
        \item An eviction made by \textsc{RandomAlg} (triggered by the previous eviction),
        \item An eviction made by \textsc{Corrector} (triggered by the previous eviction).
    \end{enumerate}
    Note that some triples may contain fewer than three evictions if certain evictions do not trigger subsequent ones. However, this is not an issue, as we only require that each triple includes at least one eviction made by \textsc{BlindOracle}.

    Consider the triples in which \textsc{BlindOracle} evicts the optimal page. By~\eqref{eq:lemma_notmuchworse_2}, the number of such triples is at most $\mathrm{OPT}$. Therefore, the number of mistakes in these triples is at most $2 \cdot \mathrm{OPT}$ (assuming both \textsc{RandomAlg} and \textsc{Corrector} make mistakes in every case).

    Now, consider the remaining triples, i.e., those in which \textsc{BlindOracle} makes a mistake. The number of such triples is $|B|$, so the total number of evictions (and thus mistakes) in these triples is at most $3|B|$.
\end{proof}

\section{Conclusion} \label{Sect:Conclusion}

In this work, we established two main results for learning-augmented online caching with next-occurrence predictions. First, we improved the competitive ratio upper bound for the \textsc{BlindOracle} algorithm. Second, we introduced a new algorithm, \textsc{AlternatingOracle}, and proved that it achieves an asymptotically superior competitive ratio compared to all previously known approaches. These results advance our understanding of how to leverage predictions in online caching systems.

Several open questions remain:
\begin{itemize}
\item Are our bounds tight, or can they be further improved?
\item Do there exist randomized algorithms with strictly better asymptotic performance than \textsc{AlternatingOracle}?
\end{itemize}

\bibliographystyle{plainnat}
\bibliography{refs} 

\appendix

\end{document}